\documentclass[conference]{IEEEtran}

\usepackage{amssymb,amsfonts,amsmath,bbm,bm}
\usepackage{dsfont}
\usepackage[keeplastbox]{flushend}
\usepackage{balance}
\usepackage{multirow}
\usepackage{kantlipsum}
\usepackage{cuted}
\usepackage{tabu}
\usepackage{hyperref}
\usepackage{enumerate,mathtools}
\usepackage{amsthm}
\usepackage[noadjust]{cite}
\usepackage{microtype}
\usepackage{color}
\usepackage{subcaption}
\usepackage{fixltx2e}
\usepackage{url}
\usepackage{pgfplots}

\newcommand{\nn}{\nonumber \\}

\newtheorem{lemma}{\textbf{Lemma}}[section]
\newtheorem{proposition}{\textbf{Proposition}}[section]
\newtheorem{theorem}{\textbf{Theorem}}[section]

\theoremstyle{definition}

\usetikzlibrary{positioning}
\pgfplotsset{compat=1.12}

\newcommand{\bW}{\bm{W}}
\newcommand{\bX}{\bm{X}}
\newcommand{\bY}{\bm{Y}}

\newcommand{\bx}{\bm{x}}

\newcommand{\cI}{\mathcal{I}}

\newcommand{\cS}{\mathcal{S}}

\newcommand{\cX}{\mathcal{X}}

\DeclareMathOperator{\cov}{\mathsf{Cov}}

\DeclareMathOperator{\gvec}{\mathsf{vec}}
\DeclareMathOperator{\vech}{\mathsf{vech}}

\DeclareMathOperator{\MMSE}{\mathsf{MMSE}}

\newcommand{\reals}{\mathbb{R}}

\let\originalleft\left
\let\originalright\right
\renewcommand{\left}{\mathopen{}\mathclose\bgroup\originalleft}
\renewcommand{\right}{\aftergroup\egroup\originalright}

\newenvironment{talign}
 {\align}
 {\endalign}
\newenvironment{talign*}
 {\csname align*\endcsname}
 {\endalign}

\newcommand*{\QEDA}{\hfill\ensuremath{\blacksquare}}

\IEEEoverridecommandlockouts

\newtheorem{hypothesis}{Hypothesis}

\begin{document}

\title{Information-theoretic limits of a multiview\\ low-rank symmetric spiked matrix model}

\author{
\IEEEauthorblockN{
	Jean Barbier\IEEEauthorrefmark{1}
	and Galen Reeves\IEEEauthorrefmark{2}	
}
\IEEEauthorblockA{\IEEEauthorrefmark{1}The Abdus Salam International Center for Theoretical Physics, Trieste, Italy.}
\IEEEauthorblockA{\IEEEauthorrefmark{2}Department of ECE and Department of Statistical Science, Duke University\thanks{This work was supported in by the NSF under Grant No. 1750362.}
}
}
\maketitle

\begin{abstract}
We consider a generalization of an important class of high-dimensional inference problems, namely spiked symmetric matrix models, often used as probabilistic models for principal component analysis. Such paradigmatic models have recently attracted a lot of attention from a number of communities due to their phenomenological richness with statistical-to-computational gaps, while remaining tractable. We rigorously establish the information-theoretic limits through the proof of single-letter formulas for the mutual information and minimum mean-square error. On a technical side we improve the recently introduced adaptive interpolation method, so that it can be used to study low-rank models (i.e., estimation problems of ``tall matrices'') in full generality, an important step towards the rigorous analysis of more complicated inference and learning models.
\end{abstract}

\IEEEpeerreviewmaketitle

\section{Introduction}

It is fair to say that in the last decade our rigorous understanding of high-dimensional Bayesian inference has experienced a tremendous improvement. Extremely rich and complex models appearing in a variety of disciplines are now under control, at least at the information-theoretic level, and important steps towards understanding algorithmic limits have been made too. Some paradigmatic examples in high-dimensional regression that have now been put on a rigorous basis are: code division multiple access \cite{korada2010} and sparse superposition codes \cite{barbier2016proof,rush2017capacity} in communication, that are closely related to compressive sensing in signal processing \cite{korada2010,reeves2016replica,barbier2016mutual,barbier2018mutual,reeves:2019d}; estimation and learning in generalized regression, including as special cases non-linear compressive sensing or the perceptron neural network \cite{barbier2019optimal}; or models of shallow and deep neural netwoks \cite{gabrie2018entropy,aubin2018committee}. A second class of models that has attracted great attention is matrix (and more generally tensor) estimation, originally introduced as simple probabilistic models for principal component analysis \cite{johnstone2001}. These papers have generated intense research studying the limits of spectral algorithms \cite{baik2005phase}, semidefinite relaxations and the sum-of-squares hierarchy \cite{amini2009,krauthgamer2015,JMLR:v17:15-160,Ma:2015:SLB:2969239.2969419}, the approximate message-passing algorithm (AMP) \cite{deshpande2014information,fletcher2018iterative}, or gradient-based methods \cite{mannelli2018marvels}. 

More closely related to the present paper, the information-theoretic limits of the minimum mean-square error (MMSE) estimator have been studied through a series of results on the validity of so-called ``replica symmetric formulas'' for the mutual information and MMSE, derived by statistical physics tools \cite{lesieur2017constrained}, and proved by various methods in \cite{korada2009exact,krzakala2016mutual,dia2016mutual,barbier2018rankone,lelarge2019fundamental,2017arXiv170910368B,BarbierM17a,el2018estimation,barbier20190}. These spiked matrix models are also intimately connected to stochastic block models for community detection \cite{deshpande2015asymptotic,lelarge2019fundamental,reeves2019geometry,mayya2019mutual}.

This leap in progress is due in particular to a balanced mixture of new (and older) tools developed mainly in statistical physics and information theory. We mention the rigorous control of AMP through state evolution \cite{donoho2009message,bayati2011dynamics} leading to algorithmic proofs such as \cite{barbier2016mutual,barbier2018mutual}; the cavity method \cite{mezard2009information,talagrand2010meanfield1}, applied in inference, e.g., in \cite{lelarge2019fundamental}; interpolation techniques developed in \cite{guerra2002thermodynamic,Guerra-2003} and later improved in the context of inference \cite{BarbierM17a,BarbierMacris2019,mourrat:2018}; information-theoretic proofs \cite{reeves2016replica}\footnote{This reference is also implicitly based on ideas with a flavor of interpolation and cavity methods.}.


The present work is part of this line of research and brings forward two main contributions: 
\begin{itemize}
\item The introduction and information-theoretic analysis, through the proof of single-letter variational formulas for the mutual information and MMSE, of a generalized spiked matrix model encompassing a number of models in the literature.
\item A key improvement of the {\it adaptive interpolation method} \cite{BarbierM17a,BarbierMacris2019}, that has proven in the past years to be one of the simplest and most versatile techniques to study high-dimensional inference problems. This novelty allows to generically treat low-rank models that were until now out of reach for the method\footnote{We note the exception of low-rank even-order symmetric tensor estimation \cite{barbier2019mutual} that can be treated thanks to non-generic symmetries of the model.}.
\end{itemize}
\section{Multiview symmetric spiked matrix model}

\subsection{Problem formulation}
Let $\bX = (X_1, \dots, X_n)^\intercal\in \mathbb{R}^{n \times d}$ be an unknown random signal-matrix whose rows are i.i.d.\ according to a distribution $p_X$ supported on a bounded subset $\cX$ of $\reals^d$. 
Consider the observation model:
\begin{subequations}
\label{eq:general_model}
 \begin{alignat}{3}
\tilde \bY & = \bX S^{1/2} +\tilde \bW, \label{eq:Ytilde}\\
\bY_\ell & = n^{-1/2} \bX B_\ell \bX^\intercal +  \bW_\ell, \quad && \ell = 1, \dots, L,
\end{alignat}
\end{subequations}
where $S\in {\cal S}_d^+$, the set of $d\times d$ semidefinite non-negative matrices, and $B_1, \dots, B_L \in \reals^{d \times d}$ are deterministically known, while $\tilde\bW \in \reals^{n \times d}$ and $\bW_1, \dots, \bW_L \in \reals^{n \times n}$ are independent noise matrices with i.i.d.\ standard Gaussian entries. Our analysis considers the regime $d,L$ 
 fixed while taking the limit $n\to +\infty$. We keep the bold notation for matrices with at least one dimension scaling with $n$. 

The observation model \eqref{eq:general_model} encompasses a number of estimation problems found in the literature. The most studied example is the spiked Wigner model when taking $S=0_d$, $L=1$ and $B_1 =\gamma \,{\rm I}_d$. Another important one is a class of stochastic blocks models studied in \cite{reeves2019geometry,mayya2019mutual}.

Note that, by stability of the Gaussian law under addition, this model is equivalent to a model where the first channel in \eqref{eq:general_model} is replaced by $K$ independent Gaussian channels of the form $\tilde \bY_k= \bX A_k +\tilde \bW_k$ where $A_k\in\mathbb{R}^{d\times m_k}$. More precisely, this equivalence follows by setting $S=\sum_{k=1}^KA_kA_k^\intercal$, and by equivalence we mean that the Le Cam distance between the two models vanishes \cite{le2012asymptotic}, which in particular implies mutual information and MMSE-wise equivalence. 

One could even further generalize the model by considering additional higher-order tensor observations of the form 
$Y^{(p)}_{j,i_1\dots i_p}\!=\! n^{\frac{1-p}{2}} \sum_{k_1, \dots, k_p}^d \!C^{(p)}_{j,k_1 \dots k_p}X_{i_1k_1}\dots X_{i_pk_p} \!+\! W^{(p)}_{j,i_1\dots i_p}$
with $C^{(p)}_j$ an order-$p$ coupling tensor and $\bW_j^{(p)}$ an order-$p$ Gaussian noise tensor, with $p\ge 3$, $1\le j\le J$, $(i_1, \ldots, i_p)\in\{1,\ldots,n\}^p$. The analysis would be more involved notation-wise but would not present any additional difficulties, and is left for future work.



\subsection{The minimum mean-square error matrix} 

Our results describe the information-theoretic limits of estimating the signal matrix $\bX$ from the collection of observations $\bY\equiv (\bY_\ell)$ and $\tilde \bY$. We focus on the asymptotic mutual information as well as a multivariate performance measure called the MMSE matrix~\cite{reeves2019geometry}, which is defined by 
\begin{talign*}
\MMSE(\bX \mid \tilde \bY,\bY) \equiv  \frac{1}{n} \sum_{i=1}^n \mathbb{E}[\cov( X_i \mid \tilde \bY,\bY) ],
\end{talign*}
where the $d \times d$ conditional covariance matrix $\cov( X_i \mid \tilde \bY,\bY)$ of the row $X_i$ of $\bX$ (which, when considered alone, is a column vector) is, more explicitly, $$ \mathbb{E}\big[(X_i - \mathbb{E}[X_i \mid \tilde \bY,\bY])(X_i - \mathbb{E}[X_i \mid \tilde \bY,\bY])^\intercal  \mid \tilde \bY,\bY\big]. $$
Note that the trace of this matrix corresponds to the usual definition of the MMSE:
\begin{talign*}
{\rm Tr}\MMSE(\bX \mid \tilde \bY,\bY)   =  \frac{1}{n} \mathbb{E}\big[ \|  \bX - \mathbb{E}[\bX \mid \tilde \bY,\bY]\|_{\rm F}^2\big].
\end{talign*}
One of the advantages of the matrix MMSE is that is provides information about the individual dimensions in the row space of $\bX$, which can be important in settings where some dimensions can be recovered while other cannot; see~\cite{reeves2019geometry,aubin2018committee}.

Our analysis leverages the matrix I-MMSE relation~\cite{payaro2011yet,reeves2018mutual}, which relates the MMSE matrix to the gradient of the mutual information: 
\begin{talign*}
\MMSE(\bX \mid \tilde \bY,\bY) =   \frac{2}{n}  \nabla_S I\big(\bX;(\tilde \bY,\bY)\big),
\end{talign*}
where $S$ is the $d \times d$ matrix appearing in \eqref{eq:Ytilde}. Note that  evaluating this gradient at $S = 0_d$ recovers the MMSE matrix associated with only the observations $\bY$.

\subsection{Statement of main results}

We provide formulas that depend only on the mutual information and MMSE associated with a $d$-dimensional estimation problem. Define 
for $S\in{\cal S}_d^+$, the functions
\begin{align*}
I_S:r\in{\cal S}_d^+ &\mapsto I(X;(S+r)^{1/2}X + W)\\
M_S:r\in{\cal S}_d^+ &\mapsto \mathbb{E} [\cov(X \mid (S+r)^{1/2}X + W)],
\end{align*}
where $X\sim p_X$, $W\sim{\cal N}(0,{\rm I}_d)$. Let $\bar \rho\equiv \mathbb{E}[XX^\intercal]$  be the $d\times d$ second moment matrix of $p_X$ and define $f : {\cal S}_d^+ \mapsto \reals$ as
\begin{align}
f(S) &\equiv{\adjustlimits \inf_{r\in {\cal S}^+_d} \sup_{q\in{\cal S}^+_{d}}} \, {\cal I}(r,q) \label{eq:f}
\end{align}
where the following ``replica symmetric potential'' function ${\cal I}(r,q)$ depends also implicitly on $(S,(B_\ell))$:
\begin{align}
	&{\textstyle {\cal I}(r,q)\!\equiv\! I_S(r)+\frac12{\rm Tr}\big[r(q-\bar \rho)\!+\!\sum_{\ell}^L\{B_\ell^\intercal \bar \rho B_\ell \bar \rho-B_\ell^\intercal q B_\ell q\}\big]}.\nonumber
\end{align}

Our results require some structure for the $(B_\ell)$  matrices:
\begin{hypothesis}[Positive coupling structure]\label{hypB} {\it The $(B_\ell)$ are s.t.}
\begin{talign*}
 \sum_{\ell=1}^L \left\{  \left( B_\ell \otimes B_\ell \right) + \left( B_\ell \otimes B_\ell \right)^\intercal  \right\} \succcurlyeq 0.
\end{talign*}
\end{hypothesis}

\begin{theorem}[Replica symmetric formulas]\label{mainThm}
Under hypothesis~\ref{hypB} the mutual information for model \eqref{eq:general_model} verifies, 
\begin{talign}
	&\lim_{n\to+\infty}{ \frac1n}I\big(\bX;(\tilde \bY,\bY)\big)=f(S)
\end{talign}
for all  $S \in {\cal S}_d^+$. Furthermore,  $f$ is concave and hence differentiable almost everywhere on ${\cal S}_d^{++}$. At each point  $S \in {\cal S}_d^{++}$ where $f$ is differentiable, the MMSE matrix verifies
\begin{align}
&\lim_{n\to+\infty} \MMSE( \bX \mid \tilde{\bY}, \bY) =  \tfrac{1}{2} \nabla f(S).
\end{align}
\end{theorem}

%


Note that this theorem implies in particular the tightness of the bounds for community detection provided in \cite{reeves2019geometry,mayya2019mutual}.

Let us comment briefly on the motivation for including the submodel \eqref{eq:Ytilde} in our analysis. Many of the problem formulations in the literature correspond to the setting $S = 0_d$. In these settings, our results characterize the exact limit of the mutual information, and using standard perturbation arguments similar to \cite{barbier2019optimal,lelarge2019fundamental}, one can then  verify the limiting behavior of the scalar MMSE associated with each $n \times n$ matrix $\bX B_\ell \bX^\intercal$.

However, if the goal it to recover the matrix $\bX$ itself, then the analysis becomes more difficult due to invariances in the problem. For example, if $\bX$ is equal in distribution to $-\bX$ and the matrices $(B_\ell)$ are symmetric, then the posterior distribution of $\bX$ given $\bY$ is also symmetric with respect to a sign change. In this case the conditional expectation is deterministically zero and so the MMSE is constant. As is argued in~\cite{reeves2019geometry} the inclusion of submodel \eqref{eq:Ytilde} provides an approach to resolve these non-identifiability issues that is both interpretive and intuitive. The basic idea is that an arbitrarily small but positive definite $S$ is sufficient to break the symmetry in the model and thus the double limit
\[
\lim_{S \to 0_d} \lim_{n \to +\infty}  \MMSE( \bX \mid \tilde{\bY}, \bY),
\]
provides a meaningful measure of performance, even if the $S = 0_d$ limit is degenerate. 

The potential function ${\cal I}$ verifies the following stationary conditions, called {\it state evolution} equations: 
\begin{talign}
 \nabla_r {\cal I}(r,q) = 0_d \Leftrightarrow q&=\bar \rho -M_S(r),\nn
\nabla_q {\cal I}(r,q) = 0_d \Leftrightarrow r&=\sum_{\ell=1}^L \big\{B_\ell qB_\ell^\intercal+B_\ell^\intercal qB_\ell\big\}\nn
&\equiv r^*(q).\label{defrstar}
\end{talign}
At a stationary point with respect to $q$ the potential is
\begin{talign}
&{\cal I}(r^*(q),q)	\!=\!I_S(r^*(q))\!+\!\frac12\sum_{\ell}^L {\rm Tr}\big[ B_\ell^\intercal (\bar \rho \!-\!q)B_\ell( \bar \rho\!-\! q )\big].\label{pot_q_only}
\end{talign}

A fact we will need later in the proof, and which is the reason for hypothesis~\ref{hypB}, is the following:
\begin{lemma}[Concavity of potential]\label{lemma:concav}
Under hypothesis~\ref{hypB} the potential $q\in{\cal S}_d^+ \mapsto {\cal I}(r,q)$ is concave.
\end{lemma}
\begin{proof}  
Using vectorization, the $q$-dependent piece of the potential $g: q \in \cS^+_d  \mapsto {\rm Tr}[rq] -  \sum_{\ell}  {\rm Tr}[ B_\ell^\intercal q B_\ell q]$ can be expressed as 
\begin{talign*}
g(q) &= \vech(q) D_d^\intercal  D_d \vech(r) \nn
&\qquad\qquad- \vech(q)^\intercal \sum_{\ell}  D^\intercal_d (B_\ell \otimes B_\ell) D_d \vech(q) 
\end{talign*}
where  $\vech(q)$ is the $d (d+1)/2 \times 1$ vector obtained by stacking the entries on or below the diagonal of $q$, and $D_d$ is the duplication matrix, i.e., 
the unique $d^2 \times d(d+1)/2$ matrix such that for any $M \in \cS_d$ (the set of symmetric $d\times d$ matrices), $\gvec(M) = D_d \vech(M)$ where $\gvec(M)$ is the $d^2 \times 1$ vector obtained by stacking the columns of $M$; see \cite[Chapter~3.8]{magnus:2007}. Thus, the Hessian of $g$ can be expressed as  
$- D_n^\intercal \sum_{\ell}\{  (B_\ell \otimes B_\ell) + (B_\ell \otimes B_\ell)^\intercal\} D_n$.	
Under hypothesis~\ref{hypB}, this matrix is negative semidefinite and therefore the concavity of the potential follows. 
\end{proof}

Formula \eqref{eq:f} can also be expressed in terms of a one-letter potential. This expression is equivalent\footnote{There is factor of two difference in the formulas that arises from the fact that \cite{mayya2019mutual} considers symmetric noise matrices.} to the one given in \cite[Theorem~1]{mayya2019mutual} for the case where $(B_\ell)$ are symmetric. 

\begin{lemma}[Single-Letter Formula] \label{lem:f_alt} 
The function $f(S)$ in \eqref{eq:f} can also be expressed as
\begin{multline*}
f(S)  = \inf_{q \in {\cal S}_d^+ } \Big\{  I_S\big(  {\textstyle \sum_{\ell=1}^L \{B_\ell q B_\ell^\intercal + B^\intercal_\ell q B_\ell\} }\big ) \\ + \textstyle \frac12{\rm Tr}\big[\sum_{\ell=1}^L\{B_\ell^\intercal (\bar \rho -q) B_\ell (\bar \rho -q) \}\big] \Big\} .
\end{multline*}
\end{lemma}
\begin{proof}
Let $r^*(q) $ be given by \eqref{defrstar}. For each $\tilde{q}$ it follows Lemma~\ref{lemma:concav} that $
\sup_{q\in {\cal S}^+_d} \cI( r^*(\tilde{q}), q) = \cI( r^*(\tilde{q}), \tilde{q} ) $. Restricting the infimum over $r$ to the image of $q \mapsto r^*(q)$ leads to an upper bound
\begin{align*}
f(S) 
&  \le   \adjustlimits \inf_{\tilde{q} \in{\cal S}^+_{d}}  \sup_{q\in {\cal S}^+_d} \cI( r^*(\tilde{q}), q) =  \inf_{\tilde{q} \in{\cal S}^+_{d}}  \cI( r^*(\tilde{q}), \tilde{q} ).
\end{align*}
Alternatively, because $\cI(r, q)$ is convex-concave,  standard duality arguments show that we can interchange the role of $r$ and $q$ in the inf sup, and this leads to a matching lower bound
\[ f(S) = \adjustlimits \inf_{q\in {\cal S}^+_d} \sup_{r\in{\cal S}^+_{d}} \, {\cal I}(r,q) \ge  \inf_{q\in {\cal S}^+_d}  {\cal I}(r^*(q) ,q). 
 \]
Comparing with \eqref{pot_q_only} this gives the stated result. 
\end{proof}

\section{Adaptive interpolation, reloaded} 

In this section we prove the first part of theorem~\ref{mainThm} using an evolution of the adaptive interpolation method \cite{BarbierM17a,BarbierMacris2019}.

\subsection{Interpolating model}
We consider a model parameterized by the {\it time} $t\in[0,1]$, a {\it perturbation} $$\epsilon=(\lambda,\eta\, {\rm I}_d)\in {\cal S}^+_d\times {\cal S}^+_d$$ with a scalar $\eta\ge 0$, and a generic {\it interpolation function} $R(t,\eta)\in {\cal S}^+_d$ verifying $R(0,\eta)=0$. Note that the interpolation function depends on $\eta$ but {\it not} on $\lambda$. We require that $\epsilon=\epsilon_n$ verifies $\|\epsilon\|_{\rm F}\le b_n \to 0_+$. Define $\bX_\ell\equiv \bX B_\ell \bX^\intercal$. The {\it interpolating model} is
\begin{subequations}
\label{eq:general_model_alt}
\begin{alignat}{3}
\tilde \bY(t,\eta) & = \bX  \big(S+\eta\, {\rm I}_d+R(t,\eta)\big)^{1/2}  +\tilde \bW ,\label{eq:general_model_alt_1}\\
\hat \bY(\lambda) & = \bX  \lambda^{1/2}  +\hat \bW ,\label{eq:general_model_alt_2}\\
\bY_\ell(t) & = \sqrt{(1-t)/n}\,\bX_\ell  + \bW_\ell, \quad \ell=1,\ldots, L.\label{eq:general_model_alt_3} 
\end{alignat}
\end{subequations}
where all the $\bW$ matrices are independent and made of i.i.d. ${\cal N}(0,1)$ entries. Let $\bY(t)\equiv (\bY_\ell(t))$. The {\it interpolating mutual information} is
\begin{talign*}
\cI_n(t,\epsilon)&\equiv \frac1n I\big(\bX;(\tilde \bY(t,\eta),\hat \bY(\lambda),\bY(t))\big).
\end{talign*}
By construction it verifies 
\begin{talign*}
	 \cI_n(0,\epsilon)&=\frac1nI\big(\bX;(\tilde \bY,\bY)\big)+O(b_n),\\
	\cI_n(1,\epsilon)&=I_S(R(1,\eta))+O(b_n).
\end{talign*}
So at $t=0$ the mutual information of the original model \eqref{eq:general_model} appears naturally (recall $\bY\equiv (\bY_\ell)$). The $O(b_n)$'s, that are uniform in $(t,\epsilon)$, are extracted using the chain rule for mutual information and the Lipschitzianity of $\epsilon\mapsto\cI_n(t,\epsilon)$, with Lipschitz constant depending only on $|\sup {\cal X}|<\infty$ (this follows from the I-MMSE relation, see \cite{barbier20190} for similar computations). We compare these values using the fundamental theorem of calculus
\begin{talign*}
\cI_n(0,\epsilon)=\cI_n(1,\epsilon)-\int_0^1  \cI_n'(t,\epsilon)dt,	
\end{talign*}
where the prime $'$ will always mean $t$-derivative. Denote the expectation with respect to the interpolating model posterior $P(\cdot\,|\,\tilde \bY(t,\eta),\hat \bY(\lambda),\bY(t))$ as $$\langle -\rangle_{t,\epsilon}\equiv \mathbb{E}[-|\tilde \bY(t,\eta),\hat \bY(\lambda),\bY(t)].$$ 

We evaluate $\cI_n'$ which is directly obtained from the matrix I-MMSE relation. Denote $\tilde \bX\in{\cal X}^{n\times d}$ a random sample from the posterior $P(\cdot\,|\,\tilde \bY(t,\eta),\hat \bY(\lambda),\bY(t))$, the concise notation $\tilde \bX_\ell \equiv \tilde \bX B_\ell\tilde \bX^\intercal$, and finally $\mathbb{E}$ the joint expectation with respect to the data $(\tilde \bY(t,\eta),\hat \bY(\lambda),\bY(t))$. Define 
\begin{talign*}
	\rho \equiv \frac1n\bX^\intercal\bX \qquad\text{and} \qquad Q\equiv \frac1n \bX^\intercal \tilde \bX,
\end{talign*}
(recall $\bar \rho\equiv \mathbb{E}[XX^\intercal]=\mathbb{E}\rho$) where $Q$ is the $d\times d$ {\it overlap} matrix. Then the mutual information time derivative reads
\begin{talign*}
2\,\cI_n'(t,\epsilon)&=\mathbb{E}\,{\rm Tr}\big[\frac1nR'(t,\eta)(\bX\!-\!\langle \tilde \bX\rangle_{t,\epsilon})^\intercal(\bX\!-\!\langle \tilde \bX\rangle_{t,\epsilon})\big]\\
& \qquad-\!\sum_{\ell}\mathbb{E}\,{\rm Tr}\big[\frac1{n^2}(\bX_\ell\!-\!\langle \tilde \bX_\ell\rangle_{t,\epsilon})^\intercal(\bX_\ell\!-\!\langle \tilde \bX_\ell\rangle_{t,\epsilon})\big]\nn
&={\rm Tr}\big[R'(t,\eta)(\bar\rho\!-\!\mathbb{E}\langle Q\rangle_{t,\epsilon})\big] \nn
&\qquad-\sum_\ell{\rm Tr}\big[\mathbb{E}[B_\ell \rho B_\ell^\intercal \rho]\!-\!\mathbb{E}\langle B_\ell Q^\intercal B_\ell^\intercal Q\rangle_{t,\epsilon}\big],
\end{talign*}
using that the average overlap $\mathbb{E}\langle Q\rangle_{t,\epsilon}$ is a symmetric matrix, as well as the two identities \footnote{Such identities are sometimes called ``Nishimori identities'' in the literature.} 
\begin{talign*}
\mathbb{E}[\bX^\intercal\langle  \tilde \bX\rangle_{t,\epsilon}]&=\mathbb{E}[\langle  \tilde \bX\rangle_{t,\epsilon}^\intercal\langle  \tilde \bX\rangle_{t,\epsilon}] \nn 
\text{and} \ \ \mathbb{E}[\bX_\ell^\intercal\langle \tilde \bX_\ell \rangle_{t,\epsilon}]&=\mathbb{E}[\langle\tilde \bX_\ell \rangle_{t,\epsilon}^\intercal\langle {\tilde \bX_\ell}\rangle_{t,\epsilon}].
\end{talign*}
Plugging this back in the fundamental theorem of calculus above, using that $\rho = \bar\rho +o_n(1)$ by the law of large numbers, and having in mind that we will soon exploit a concentration property for $Q$, we obtain the {\it sum rule}:
\begin{talign}
&\frac1nI\big(\bX;(\tilde \bY,\bY)\big)=I_S(R(1,\eta))+o_n(1)\label{sumrule}\\
&\qquad +\frac12\int_{0}^1dt\big\{ {\rm Tr}\big[R'(t,\eta)(\mathbb{E}\langle Q\rangle_{t,\epsilon}-\bar \rho)\big] + {\cal R}^{\rm fluc}(t,\epsilon)\nn
&\qquad +\sum_\ell{\rm Tr}\big[B_\ell^\intercal \bar \rho B_\ell \bar \rho - B_\ell^\intercal \mathbb{E}\langle Q\rangle_{t,\epsilon} B_\ell \mathbb{E}\langle Q\rangle_{t,\epsilon} \big]\big\},\nonumber
\end{talign}
where $o_n(1)$ vanishes uniformly in $(t,\epsilon)$ and the {\it remainder} is
\begin{talign*}
{\cal R}^{\rm fluc}(t,\epsilon)&\equiv \sum_{\ell=1}^L {\rm Tr}\big[B_\ell\mathbb{E} \big\langle Q^\intercal B_\ell^\intercal(\mathbb{E}\langle Q\rangle_{t,\epsilon}- Q)\big\rangle_{t,\epsilon}\big].
\end{talign*}
This remainder is small if the overlap fluctuations w.r.t. the measure $\mathbb{E}\langle -\rangle_{t,\epsilon}$ are small. This is shown as follows. Consider a perturbation $\lambda\in {\cal D}_{n}\equiv \{\lambda\in {\cal S}_d^+:\lambda_{kk'}\in (s_n,2s_n) \ \text{if}\  k'\neq k, \lambda_{kk}\in (2 ds_n,(2d+1)s_n)\}$ for a sequence $(s_n)$ of positive numbers that accumulate to $0$. Because the interpolating model depends on $\lambda$ {\it only} through the explicit dependence in channel \eqref{eq:general_model_alt_2}, the results of \cite{barbier2019overlap} {\it directly apply}. There it is shown that quite generically the overlap concentrates under $\mathbb{E}\langle -\rangle_{t,\epsilon}$, in $\lambda$-average, as long as the ``free energy'' $F_n$ concentrates: 
\begin{talign*}
F_n\equiv -\frac1n\ln \int dP_X(\tilde \bx)P(\tilde \bx|\tilde \bY(t,\eta),\hat \bY(\lambda),\bY(t))\,.
\end{talign*}
Notice that $\mathbb{E}F_n\!=\!\frac1n h(\tilde \bY(t,\eta),\hat \bY(\lambda),\bY(t))$ is the differential entropy of the data.
\begin{proposition}[Overlap concentration \cite{barbier2019overlap}] \label{prop:overlapConc} Suppose 
\begin{talign*}
\mathbb{E}\big[(F_n-\mathbb{E} F_n)^2\big]= O(\frac{1}{n}).	
\end{talign*}
Let $\mathbb{E}_\lambda[-] \!\equiv\! {\rm Vol}({\cal D}_{n})^{-1}\!\int_{{\cal D}_{n}}\!d\lambda[-]$. Then we have, uniformly in $(t,\eta)$ and the choice of $R$,
\begin{talign*}
\mathbb{E}_\lambda\mathbb{E}\big\langle \|Q-\mathbb{E}\langle Q\rangle_{t,\epsilon}\|_{\rm F}^2\big\rangle_{t,\epsilon} = O((s_n^4n)^{-1/6}).
\end{talign*}
\end{proposition}
Proving $\mathbb{E}[(F_n\!-\!\mathbb{E} F_n)^2]= O(\frac1n)$ whenever $P_X=p_X^{\otimes n}$, i.e., decouples over the rows, is standard. E.g., one can slightly adapt the proof of concentration found in \cite{barbier20190,BarbierMacris2019}. Equipped with proposition~\ref{prop:overlapConc}  a simple application of the Cauchy-Schwarz inequality yields $\mathbb{E}_{\lambda} |{\cal R}^{\rm fluc}(t,\epsilon)|=o_n(1)\to 0_+$ uniformly in $(t,\eta)$ as $n\to+\infty$, when choosing an appropriate sequence $(s_n)$ that vanishes (but not too fast, i.e., $s_n=\omega(n^{-1/4})$). Then by Fubini's theorem we have, uniformly in $\eta$ and $R$,
\begin{talign}
\mathbb{E}_{\lambda}\int_0^1 dt\, {\cal R}^{\rm fluc}(t,\epsilon)=o_n(1).\label{smallR}	
\end{talign}

\subsection{Upper bound on mutual information} 
We now exploit our freedom of choice of the interpolation function $R$ and set its derivative to a constant $R'(t,\eta)=r\in {\cal S}^+_d$, so that $R(1,\eta)=r$ too. Under this choice, averaging the sum rule \eqref{sumrule} w.r.t. $\lambda\in {\cal D}_n$ and using \eqref{smallR} yields
\begin{talign*}
&\frac1nI\big(\bX;(\tilde \bY,\bY)\big)=\mathbb{E}_\lambda \int_{0}^1dt\,{\cal I}(r,\mathbb{E}\langle Q\rangle_{t,\epsilon})+o_n(1) \ \forall \ r\in {\cal S}^+_d.
\end{talign*}
Therefore we obtain the upper bound:
\begin{align}
&\limsup_{n\to+\infty}\frac1nI\big(\bX;(\tilde \bY,\bY)\big)\le {\adjustlimits \inf_{r\in {\cal S}^+_d} \sup_{q\in{\cal S}^+_{d}}}\,{\cal I}(r,q).\label{upper}
\end{align}

\subsection{Lower bound on mutual information} 
Let a new perturbation $$\tilde\epsilon\equiv (0_d,\eta\,{\rm I}_d).$$ This time we select $R$ as the unique solution of the following differential equation with initial condition $R(0,\eta)=0$:
\begin{talign}
R'(t,\eta)&=	\sum_{\ell=1}^L \big\{B_\ell \mathbb{E}\langle Q\rangle_{t,\tilde\epsilon}B_\ell^\intercal+B_\ell^\intercal \mathbb{E}\langle Q\rangle_{t,\tilde\epsilon}B_\ell\big\} \nn
&= r^*(\mathbb{E}\langle Q\rangle_{t,\tilde\epsilon}).\label{ODE}
\end{talign}
We used definition \eqref{defrstar} for $r^*$. With the notation $\mathbb{E}\langle Q\rangle_{t,\tilde\epsilon}$ we emphasize that we consider an expectation of the overlap {\it along an interpolation path in which we set the perturbation $\lambda$ to the all zeros matrix $0_d$ while the other perturbation $\eta \,{\rm I}_d$ is unchanged}. The expected overlap $\mathbb{E}\langle Q\rangle_{t,\tilde\epsilon}$ is a function of $R(t,\eta)$ so the equation above is an ODE. $\mathbb{E}\langle Q\rangle_{t,\tilde\epsilon}$ is (component-wise) Lipschitz in $R$ with Lipschitz constant depending only on $|\sup{\cal X}|$ and $n$ that are both finite (this is seen using relations found, e.g., in \cite{reeves2018mutual}), so by the Cauchy-Lipschitz theorem this first order ODE admits a unique global ${\cal C}^1$ solution $R_n^*(t,\eta)$, which importantly is independent of $\lambda$. Using this solution the derivative of the mutual information is
\begin{talign*}
2\,\cI_n'(t,\epsilon)&=-\sum_\ell^L {\rm Tr}\big[B_\ell^\intercal (\bar \rho - \mathbb{E}\langle Q\rangle_{t,\tilde\epsilon}) B_\ell (\bar \rho - \mathbb{E}\langle Q\rangle_{t,\tilde\epsilon})\big]\nn
&\qquad+ {\cal R}^{\rm mis}(t,\epsilon)- {\cal R}^{\rm fluc}(t,\epsilon),
\end{talign*}
where a new remainder appeared (here $\times$ is the matrix product):
\begin{talign*}
{\cal R}^{\rm mis}(t,\epsilon)&\equiv \sum_{\ell}^L {\rm Tr} \big[ B_\ell^\intercal(\mathbb{E}\langle Q\rangle_{t,\epsilon}- \mathbb{E}\langle Q\rangle_{t,\tilde\epsilon} )B_\ell\nn
&\qquad\qquad\qquad\qquad\times (\mathbb{E}\langle Q\rangle_{t, \epsilon} - \mathbb{E}\langle Q\rangle_{t, \tilde\epsilon} )
\big].
\end{talign*}
The remainder $|{\cal R}^{\rm mis}(t,\epsilon)|$ is small if the mismatch $|\mathbb{E}\langle Q\rangle_{t,\epsilon}- \mathbb{E}\langle Q\rangle_{t,\tilde\epsilon}|$ due to the different choices of interpolation paths (that differ in the perturbation $\lambda$) is small. The purpose of the other perturbation $\eta\, {\rm I}_d$, which is a novelty w.r.t. the usual adaptive interpolation method, is to control this remainder: 
\begin{proposition}[Interpolation paths mismatch]\label{prop:Rmis} Let $\delta>0$, $\lambda\in{\cal D}_n$ and $\mathbb{E}_\eta[-]\equiv \frac1\delta\int_0^\delta	d\eta[-]$. Then, uniformly in $(\lambda,R)$,
\begin{talign*}
  \mathbb{E}_\eta\int_0^1dt\,{\cal R}^{\rm mis}(t,\epsilon) = O(s_n/\delta).
\end{talign*}
\end{proposition}
\begin{proof}
Notice $\mathbb{E}\langle Q\rangle_{t,\epsilon}- \mathbb{E}\langle Q\rangle_{t,\tilde\epsilon} \succcurlyeq 0_d$ or, said differently, the matrix MMSE is a decreasing function (w.r.t. the Loewner partial order) of the signal-to-noise matrix $\lambda$ \cite[Prop.~5]{rioul:2011}. 
Using Cauchy-Schwarz and ${\rm Tr}A\ge \|A\|_{\rm F}$ for $A\succcurlyeq 0_d$,
\begin{talign}
	|{\cal R}^{\rm mis}(t,\epsilon)|&\le \|B_\ell^\intercal(\mathbb{E}\langle Q\rangle_{t,\epsilon}- \mathbb{E}\langle Q\rangle_{t,\tilde\epsilon} )B_\ell\|_{\rm F}\nn
	&\qquad\times{\rm Tr}[\mathbb{E}\langle Q\rangle_{t,\epsilon}- \mathbb{E}\langle Q\rangle_{t,\tilde\epsilon}] .\label{toAveta}
\end{talign}	
Now notice that by the I-MMSE matrix relation \cite{reeves2018mutual} we have ${\rm Tr}[({\rm I}_{d}+ \frac{d}{d\eta} R_n^*(t,\eta))(\mathbb{E}\langle Q\rangle_{t,\epsilon}- \mathbb{E}\langle Q\rangle_{t,\tilde\epsilon})] = 2\frac{d}{d\eta}(\cI_n(t,\tilde\epsilon)-\cI_n(t,\epsilon))$. Moreover we have $\frac{d}{d\eta} R_n^*(t,\eta)\succcurlyeq 0_d$. This comparison inequality is obtained using a comparison inequality similar to \cite[Lemma~5]{lasota:1970}. We refer to \cite{reeves:2020} where this step is proven in full details

As the trace of a product of non-negative definite matrices is non-negative,
\begin{talign*}
{\rm Tr}[\mathbb{E}\langle Q\rangle_{t,\epsilon}- \mathbb{E}\langle Q\rangle_{t,\tilde\epsilon}] \le  2\frac{d}{d\eta}(\cI_n(t,\tilde\epsilon)-\cI_n(t,\epsilon)).
\end{talign*}
Therefore averaging \eqref{toAveta} over $\eta\in(0,\delta)$ gives 
\begin{talign*}
\mathbb{E}_\eta|{\cal R}^{\rm mis}(t,\epsilon)| &\le \frac C\delta\int_0^\delta d\eta\frac{d}{d\eta}(\cI_n(t,\tilde\epsilon)-\cI_n(t,\epsilon))\\
&=\frac C\delta \big\{\cI_n\big(t,(0_d,\delta \,{\rm I}_d)\big)-\cI_n\big(t,(\lambda,\delta \,{\rm I}_d)\big) \nn
&\qquad-\cI_n\big(t,(0_d,0_d)\big)+\cI_n\big(t,(\lambda,0_d)\big)\big\}
\end{talign*}
for some constant $C>0$ dependent only on $|\sup {\cal X}|$ and $(B_\ell)$. By Lipschitzianity of $\lambda\mapsto\cI_n(t,(\lambda,A))$, with Lipschitz constant depending only on $|\sup {\cal X}|$, we get $\mathbb{E}_\eta|{\cal R}^{\rm mis}(t,\epsilon)| = O(\|\lambda\|_{\rm F}/\delta)$. Finally using $\|\lambda\|_{\rm F}=O(s_n)$ and Fubini's theorem to switch $t$ and $\eta$ averages ends the proof.
\end{proof}

The function $I_S(R)$ is concave in ${\cal S}_d^+$ \cite{payaro2011yet,reeves2018mutual}. Then, because $R(1,\eta)=\int_0^1dt\, R'(t,\eta)$,
\begin{talign*}
	&I_S(R(1,\eta))\ge \int_0^1 dt\, I_S(R'(t,\eta)).
\end{talign*}
Let $\delta_n=o_n(1)$ such that $s_n/\delta_n =o_n(1)$. With the above inequality combined with proposition~\ref{prop:Rmis} and \eqref{smallR} (and the uniformity of the bounds), averaging the sum rule \eqref{sumrule} with respect to $(\lambda,\eta)\in{\cal D}_n\times(0,\delta_n)$ gives (recall definition \eqref{defrstar})
\begin{talign*}
\frac1nI\big(\bX;(\tilde \bY,\bY)\big)&\ge \mathbb{E}_{\eta,\lambda}\int_{0}^1dt\big\{I_S(r^*(\mathbb{E}\langle Q\rangle_{t,\tilde\epsilon}))\nonumber\\
&\hspace{-2cm} +\frac12 \sum_\ell^L {\rm Tr}\big[B_\ell^\intercal (\bar \rho - \mathbb{E}\langle Q\rangle_{t,\tilde\epsilon}) B_\ell (\bar \rho - \mathbb{E}\langle Q\rangle_{t,\tilde\epsilon})\big]\big\}+o_n(1)\nn
&\hspace{-1.5cm}=\mathbb{E}_{\eta,\lambda}\int_{0}^1dt \, {\cal I}\big(r^*(\mathbb{E}\langle Q\rangle_{t,\tilde\epsilon}),\mathbb{E}\langle Q\rangle_{t,\tilde\epsilon}\big) +o_n(1)
\end{talign*}
where we used \eqref{pot_q_only} in order to identify the potential. Thanks to our choice of interpolation path \eqref{ODE} we have the identity 
\begin{talign*}
& {\cal I}\big(r^*(\mathbb{E}\langle Q\rangle_{t,\tilde\epsilon}),\mathbb{E}\langle Q\rangle_{t,\tilde\epsilon}\big) = \sup_{q\in{\cal S}^+_{d}} \, {\cal I}\big(r^*(\mathbb{E}\langle Q\rangle_{t,\tilde\epsilon}),q\big).
\end{talign*}
Indeed, by lemma~\ref{lemma:concav} the function $q\in{\cal S}_d^+\mapsto {\cal I}(r,q)$ is concave under hypothesis~\ref{hypB}. The $q$ stationary condition $$\nabla_q {\cal I}(r,q=\mathbb{E}\langle Q\rangle_{t,\tilde\epsilon})=0_d \quad \Rightarrow \quad r=r^*(\mathbb{E}\langle Q\rangle_{t,\tilde\epsilon}),$$ thus the claimed identity. Therefore
\begin{talign*}
\frac1nI\big(\bX;(\tilde \bY,\bY)\big)\!&\ge\! \mathbb{E}_{\eta,\lambda}\int_{0}^1dt \, \sup_{q} \, {\cal I}\big(r^*(\mathbb{E}\langle Q\rangle_{t,\tilde\epsilon}),q\big) \!+\!o_n(1)\nn
&\ge\! \inf_{r\in{\cal S}^+_{d}}\sup_{q\in{\cal S}^+_{d}} \, {\cal I}(r,q) \!+\!o_n(1).
\end{talign*}
Finally, taking the $\liminf_{n\to + \infty}$ yields the converse bound of \eqref{upper}, and thus ends the proof of the mutual information replica symmetric variational formula. $\QEDA$

\section{Asymptotic MMSE matrix}\label{sec:MMSE}

In this section we prove the statement about the MMSE matrix in theorem~\ref{mainThm}. Let $f$ be given as in \eqref{eq:f} and define $f_n :{\cal S}_d^{+} \mapsto \reals$ according to 
 \begin{talign*}
 f_n(S) \equiv   \frac{1}{n} I\big(\bX; (\tilde \bY,\bY)\big).	
 \end{talign*}
%
By the matrix I-MMSE relation,  $f_n$ is differentiable on ${\cal S}_d^{++}$ with gradient 
\begin{talign*}
\nabla f_n(S)= \MMSE( \bX \mid \tilde{\bY}, \bY).
\end{talign*}
Meanwhile, $f$ is concave because it is the minimum of a family of concave functions and thus it is differentiable almost everywhere. Similar to \cite[Appendix A.3]{reeves2019geometry}, we will show that pointwise convergence of $f_n$ to $f$ combined with the concavity of $f_n$ implies convergence of the gradients everywhere $f$ is differentiable.


To proceed, fix any point $S\in {\cal S}_d^{++}$ such that $f$ is differentiable. By Griffiths' lemma \cite[pg.~25]{talagrand2010meanfield1}  and the pointwise convergence of $f_n$ to $f$,  the directional derivates satisfy 
\begin{align*}
\limsup_{n \to \infty} {\rm Tr}( T\nabla f_n(S) )  =  {\rm Tr}(T \nabla f(S) ) \
\end{align*}
for all $T \in {\cal S}_d$. Moreover both $\| \nabla f_n(S)\|$ and $\|\nabla f(S)\|$ are bounded in terms a constant $C$ that depends only on the support of $p_X$. Consequently, pointwise convergence of the mapping $T \mapsto {\rm Tr}(T  \nabla  f_n(S))$ on $T \mapsto {\rm Tr}(T  \nabla f(S))$ implies uniform convergence on any compact subset of ${\cal S}_d$, and we have
\begin{align*}
\MoveEqLeft \limsup_{n \to \infty} \| \nabla f_n(S)  - \nabla f(S)\| \\
& = \limsup_{n \to \infty}  \sup_{T \in {\cal S}_d \, : \, \|T\| \le 1} {\rm Tr}( T ( \nabla f_n(S)  - \nabla f(S))  = 0.
\end{align*}
We note that for points where $f$ is not differentiable (i.e., the optimization in \eqref{eq:f} does not have a unique solution), this argument can be adapted to provide lower and upper bounds on the MMSE matrix in terms of the subdifferential of $f$. 

\section{Conclusion and perspectives}

We characterized the information-theoretic limits of a class of multiview spiked matrix models. Our analysis both unifies and significantly extends the existing body of work. One important consequence of our results is to establish the tightness of the bounds obtained previously for community detection with correlated degree-balanced stochastic block models~\cite{reeves2019geometry,mayya2019mutual}.

The advances in this paper are made possible by a novel modification of the adaptive interpolation method~\cite{BarbierM17a,BarbierMacris2019}. At a high level, this modification provides a decoupling between two main components of the method, namely the \emph{interpolation path} and the \emph{perturbation} (used for proving concentration of the overlap), and thus completely bypasses a number of non-trivial technical issues that previously limited the scope of the method. As a consequence, the method can be applied generically to a larger class of inference problems.


\bibliographystyle{IEEEtran}
\bibliography{refs}

\end{document}